\documentclass[12pt]{amsart}

\usepackage[dvipsnames]{xcolor}

\usepackage[utf8]{inputenc}
\usepackage{amsfonts}
\usepackage{latexsym}
\usepackage{amssymb}
\usepackage{amsmath}
\usepackage{enumerate}
\usepackage{mathrsfs}
\usepackage{todonotes}
\usepackage{tinos}

\usepackage{hyperref} 
\hypersetup{
    colorlinks,
    linkcolor=CadetBlue,
    citecolor=CadetBlue,
    urlcolor=CadetBlue
}
\usepackage[normalem]{ulem}

\usetikzlibrary{shapes.misc}

\tikzset{cross/.style={cross out, draw=black, minimum size=2*(#1-\pgflinewidth), inner sep=0pt, outer sep=0pt},
cross/.default={5pt}}
\usepackage[bottom]{footmisc}

\usepackage[left=1.8cm, top=2.5cm,bottom=2.5cm,right=1.8cm]{geometry}

\makeatletter
\renewcommand{\email}[2][]{%
  \ifx\emails\@empty\relax\else{\g@addto@macro\emails{,\space}}\fi%
  \@ifnotempty{#1}{\g@addto@macro\emails{\textrm{(#1)}\space}}%
  \g@addto@macro\emails{#2}%
}
\makeatother

\newtheorem{thm}{Theorem}[section]

\newtheorem{lemma}[thm]{Lemma}

\newtheorem{proposition}[thm]{Proposition}

{
\theoremstyle{definition}

}
\usepackage{nomencl}
\makenomenclature

\newcommand{\ba}{\begin{eqnarray*}}
\newcommand{\ea}{\end{eqnarray*}}

\begin{document}


\title[]{Ancestral reproductive bias in 
branching processes
}


\author{David Cheek}
\address{David Cheek\\
Department of Radiology, Harvard Medical School, 25 Shattuck Street
Boston, MA 02115, 617-432-1000, USA}
\author{Samuel G.\ G.\ Johnston}
\address{Samuel G.\ G.\ Johnston\\
Department of Mathematical Sciences, University of Bath, Claverton Down, Bath, BA2 7AY, UK}

\begin{abstract}



Consider a branching process with a homogeneous reproduction law. Sampling a single cell uniformly from the population at a time $T > 0$ and looking along the sampled cell's ancestral lineage, we find that the reproduction law is heterogeneous - the expected reproductive output of ancestral cells on the lineage from time $0$ to time $T$ continuously increases. This `inspection paradox' is due to sampling bias, that cells with a larger number of
offspring are more likely to have one of their descendants sampled by virtue of their prolificity, and the bias's strength grows with the random population size and/or the sampling time $T$.  Our main result explicitly characterises the evolution of reproduction rates and sizes along the sampled ancestral lineage as a mixture of Poisson processes,  which simplifies in special cases. The ancestral bias helps to explain recently observed variation in mutation rates along lineages of the developing human embryo. 

\end{abstract}

\maketitle


\section{Introduction} \label{sec:intro}

\subsection{Uniform ancestral lineages}
Your ancestors’ reproductive behaviour is a biased representation of the historical population. For example every one of your ancestors had children, yet many people have none; and Genghis Khan is more likely to be found among your ancestors than a particular 12th century monk. The concept holds generally for any biological population. Individuals which are ancestral to a random sample have a statistically greater reproductive output than other individuals. While this bias resembles ‘survival of the fittest’, it doesn’t have to be a consequence of Darwinian selection acting on a heterogeneous population. Biased reproduction on ancestral lineages is also a feature of homogeneous populations.

\begin{figure}[h!]
\centering
\begin{tikzpicture}[xscale=1,yscale=1]
\draw[gray] (0,-3.7) -- (0,3.8);
\draw[gray] (10,-3.7) -- (10,3.8);
\node at (0,-4)   (a) {Time $0$};
\node at (10,-4)   (a) {Time $T$};
\draw[thick] (0,0) -- (2,0);
\draw[thick] (2,-2) -- (2,2);
\draw[thick] (2,-2) -- (3.8,-2);
\draw[thick] (2,0.2) -- (4.4,0.2);
\draw[thick] (2,2) -- (3.65,2) -- (3.65,3.1) -- (3.65, 1.3) -- (5.89,1.3) -- (5.89, 1.7) -- (6.88,1.7);
\draw[thick] (5.89,1.3) -- (5.89,0.7) -- (8.3,0.7);
\draw[thick] (10,0.3) -- (8.3,0.3) -- (8.3,0.9) -- (9.7,0.9);
\draw[thick] (10,1.1) -- (9.7,1.1) -- (9.7,0.7) -- (10,0.7);
\draw[thick] (9.7,0.834) -- (10,0.834);
\draw[thick] (3.65,1.9) -- (4.67,1.9);
\draw[thick] (3.65,2.3) -- (5.7,2.3);
\draw[thick] (3.65,3.1) -- (4.98,3.1) -- (4.98,3.6) -- (7.5,3.6) -- (7.5,3.8) -- (7.5,3.21) -- (9.6,3.21);
\draw[thick] (10,3.4) -- (9.6,3.4) -- (9.6,3) -- (10,3);
\draw[thick] (7.5,3.8) -- (9.4,3.8);
\draw[thick] (6.88,2.2) -- (6.88,1.5) -- (9.1,1.5);
\draw[thick] (6.88,2.2) -- (8.1,2.2) -- (8.1,2.3) -- (8.1,1.8) -- (10,1.8);
\draw[thick] (8.1,2.3) -- (10,2.3);
\draw[thick] (4.98, 3.1) -- (4.98, 2.7) -- (7.3,2.7) -- (7.3,2.8) -- (7.3,2.5) -- (9.33,2.5);
\draw[thick] (7.3,2.8) -- (10,2.8);
\draw[thick] (3.8,-1.8) -- (3.8,-1.1);
\draw[thick] (3.8,-1.8) -- (5.9,-1.8);
\draw[thick] (3.8,-1.1) -- (3.8,-2.2);
\draw[thick] (3.8,-1.1) -- (4.3,-1.1);
\draw[thick] (4.3,-0.8) -- (4.3,-1.6) -- (6.1,-1.6) -- (6.1,-1.9) -- (6.1,-1.27) -- (7.5, -1.27);
\draw[thick] (6.1,-1.9) -- (8.8,-1.9);
\draw[thick] (4.3,-0.8) -- (7.88,-0.8);
\draw[thick] (7.88,-0.3) -- (7.88,-1.2);
\draw[thick] (7.88,-0.3) -- (10,-0.3);
\draw[thick] (7.88,-1.2) -- (9.1,-1.2);
\draw[thick] (7.88,-0.9) -- (8.88,-0.9);
\draw[thick] (3.8,-2.2) -- (4.9,-2.2); 
\draw[thick] (10,-0.7) -- (9.1,-0.7) -- (9.1,-1.5) -- (10,-1.5);
\draw[thick] [line width=1.2mm, CadetBlue] (0,0) -- (2,0) -- (2,2) -- (3.65,2) -- (3.65, 1.3) -- (5.89,1.3) -- (5.89,0.7) -- (8.3,0.7) -- (8.3,0.9) -- ( 9.7,0.9) -- (9.7,0.7) -- (10,0.7);

\draw[thick] [line width=1.2mm, CadetBlue] (0,-2.9) -- (10,-2.9);

\draw [fill, white] (2,-2.9) circle [radius=0.1];
\draw [very thick, black] (2,-2.9) circle [radius=0.1];
\node at (2,-3.3)   (a) {$3$};

\draw [fill, white] (3.65,-2.9) circle [radius=0.1];
\draw [very thick, black] (3.65,-2.9) circle [radius=0.1];
\node at (3.65,-3.3)   (a) {$4$};

\draw [fill, white] (5.89,-2.9) circle [radius=0.1];
\draw [very thick, black] (5.89,-2.9) circle [radius=0.1];
\node at (5.89,-3.3)   (a) {$2$};

\draw [fill, white] (8.3,-2.9) circle [radius=0.1];
\draw [very thick, black] (8.3,-2.9) circle [radius=0.1];
\node at (8.3,-3.3)   (a) {$2$};

\draw [fill, white] (9.7,-2.9) circle [radius=0.1];
\draw [very thick, black] (9.7,-2.9) circle [radius=0.1];
\node at (9.7,-3.3)   (a) {$3$};
\end{tikzpicture}

\caption{A single lineage is sampled at time $T$ from a continuous-time Galton-Watson tree, and the reproduction times and sizes along the uniform lineage are indicated below the tree. Large reproduction events occur more frequently along the uniform ancestral lineage than their rate in the underlying population.}
\label{fig:sample}
\end{figure}
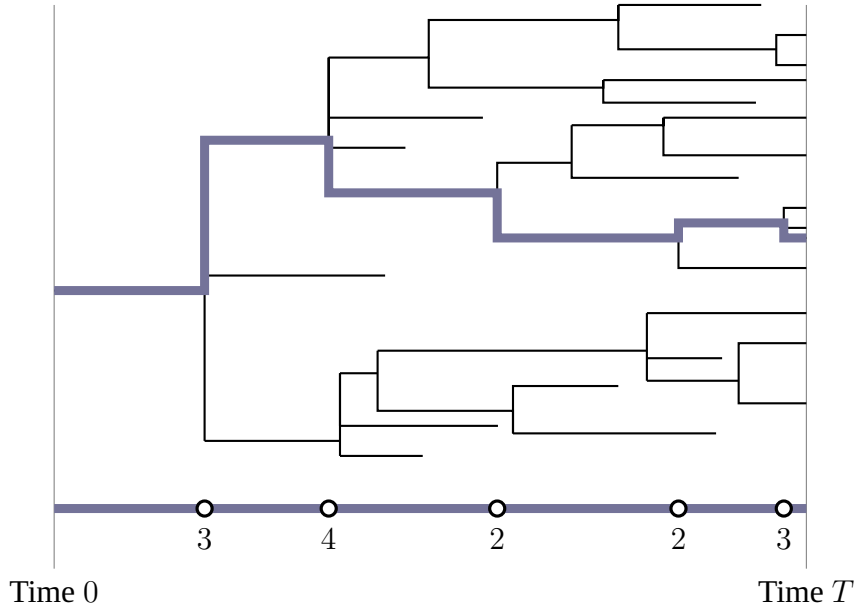

With a view to concreteness, consider a continuous-time Bienaym\'e-Galton-Watson (BGW) branching process with $N_t$ cells alive at time $t$: We begin with $N_0=1$ cell and each cell independently reproduces at rate $r$, to be replaced by $k \geq 0$ offspring with probability $p_k$ (here, $\sum_{k=0}^\infty p_k = 1$). So the total rate at which a cell is replaced by exactly $k$ children is given by $rp_k$.  Conditioning on the event $\{N_T \geq 1\}$ that there is at least one cell alive at time $T$, sample a single cell from the population at time $T$. Our object of study is the ancestral lineage of this sampled cell, which we call a `uniform ancestral lineage' (see Figure \ref{fig:sample}). What can be said about the reproduction events along a uniform ancestral lineage?

Naively, one might expect that reproduction events occur on this lineage as they do in the population, that is according to a Poisson process of rate $r$ and with each event having size $k$ with probability $p_k$. However this cannot be the case, for it is impossible that the lineage sees any reproduction event of size $k=0$. It is thus clear that somehow the reproduction rates along a uniform ancestral lineage must differ from that of the population. In fact, this absence of size $0$ reproduction events is part of a wider phenomenon of ancestral reproductive bias. Along a uniform ancestral lineage, reproduction events of size $k > 1$ tend to happen a rate greater than their rate $rp_k$ in the underlying population. 

The bias is related to the well-known inspection paradox of renewal theory, which loosely says that for a person waiting at a bus stop,  the time gap between the previous and next bus is expected to be longer than the typical time gap between buses.  The rationale is that longer time windows between buses make up a disproportionate amount of total time and are thus more likely to be sampled by the man than shorter windows are. This inspection paradox translated to branching processes says that a randomly sampled cell has a lifespan which is expected to be longer than typical cells in the population. Our interest however is not the sampled cell itself but rather its ancestors, who by contrast are expected to have shorter than typical lifespans - cells which are replaced more quickly by a larger number of offspring are more likely to have one of their descendants sampled, and are therefore more likely to land on the uniform ancestral lineage.

\subsection{Ancestral reproductive bias in large populations}

The magnitude of ancestral bias for large population sizes can be obtained via a simple heuristic argument: In a size $n$ population the total rate of size $k$ reproduction is $nrp_k$, and due to exchangeability, a child of a particular size $k$ reproduction event lands on a uniform ancestral lineage with probability $k/(n+k-1)$. Thus the rate of size $k$ reproduction events on a uniform ancestral lineage is
\begin{equation}\label{heuristic}
nrp_k\times\frac{k}{n+k-1}\approx rkp_k
\end{equation}
for large $n$, which is $k$-fold faster than the natural rate $rp_k$ in the population.

This biased ancestral reproduction has been observed by many authors in a variety of settings.
The idea goes back at least as far as \cite{CRW}, but perhaps the first explicit place this $rkp_k$ formula appears is in the work of Georgii and Baake \cite{Georgii}, who studied the uniform ancestral lineage of supercritical multitype BGW branching processes. In the large time limit, they found that the ancestral type follows a Markov chain along the lineage which, compared to the transition rates of typical individuals in the population, exhibits a bias towards more prolific types. By embedding a single-type BGW process into a multitype BGW process with countably many types (the exact construction involves declaring a particle to have type $k$ if it is destined to have $k$ offspring at the end of its lifetime), we recover the $rkp_k$ rates along the ancestral lineage.

Bansaye et al.\ \cite{BDMT} studied Markov processes indexed by supercritical Galton-Watson trees, and were able to describe the movement and reproductive behaviour along the uniform ancestral lineage for an infinite sampling time. Their result specialised to a homogeneous law of reproduction states that size $k$ birth events occur along the uniform ancestral lineage at rate $rkp_k$. We also note Marguet \cite{marguet}, who obtained similar results for branching populations starting with a large numbers of particles.

Finally,  let us mention work by the second author and coauthors \cite{HJR, johnston} have developed a framework to study questions surrounding the joint ancestry of $k$ particles sampled uniformly from continuous-time Galton-Watson trees. The main tool in their approach involves a spine change of measure associated with a size-biased version of the entire population, under which the coalescence probabilities associated with a uniform sample take a very tractable form. Here the reproduction rates and probabilities along a uniform ancestral lineage from a size biased tree are exactly given by $rkp_k$ \cite[Section 4.6]{HJR}.

For further literature on the biasing effect in uniform ancestral lineages, we refer the reader to \cite{HH, HR, kallenberg, biggins, samuels}, as well as to \cite{BDMT} and the references therein for a more detailed discussion. For further literature concerned with uniform samples of one or more particles from BGW processes, we refer the reader to \cite{Ath12a, Ath12b, Ath16, BLG, buhler, durrett, GH, johnston, JL, lambert, le, oconnell, zubkov}.

\subsection{Finite populations}

While the work in \cite{BDMT, marguet, Georgii, HJR, johnston} outlined above provides a neat approximation of reproduction along a uniform ancestral lineage when the population size is either large or size biased,  and often requires the assumption of supercriticality,  in the present article we allow for critical and subcritical regimes too, and with a view to real-world application we are interested in exact descriptions of ancestral behaviour for finite populations.

In contrast to the aforementioned limiting result,  we find that a population with a homogeneous reproductive law may exhibit a heterogeneous reproductive law along sampled ancestral lineages. Notably for a supercritical population which grows from one to many cells, the reproduction sizes and rates along a uniform ancestral lineage grow too, in the sense of stochastic dominance, between those given by the law of reproduction in the population and the size-biased version of that law.  There are implications for phylogenetics - it is not only the case that ancestors to samples from populations are statistically unusual, but moreover that variation among these ancestral individuals may give a misleading impression of variation in the historical population.

The visibility of ancestral reproduction in typical data should not be overstated however. Aside from human genealogical records,  biological populations do not usually provide a neat list of reproduction times and sizes along ancestral lineages.  Instead some information on reproduction may be recorded by mutations.  For an example consider the population of cells that make up a human body (here `reproduction' is cell division).  Recent studies \cite{Coorens,Park} have sequenced DNA of adult human cells in order to look backwards in time to the zygote, drawing phylogenetic trees of the developing human embryo that apparently depict the first cell divisions of embryogenesis.   These studies inferred variation in the mutation rate per cell division, in particular that the mutation rate was elevated for the first few divisions compared to later.  But why the mutation rate per division should vary is unclear. We explore a parsimonious explanation in terms of a branching process model of constant cell division and mutation rates, which due to ancestral bias qualitatively predicts the observation.

\subsection{Overview}
We now give a brief overview of the article:
\begin{itemize}
\item In Section 2, we present the most general result of the paper, which is a description of reproduction times and sizes along the ancestral lineage in terms of a mixture of Poisson point processes.  The proof of this result is in the spirit of classical spine arguments, where a relationship is established between the uniform ancestral lineage and another lineage generated by traversing the population's tree forwards in time.
\item In Section 3, we determine the total number of reproduction events along the uniform ancestral lineage.  The exact result simplifies for special cases of birth-death and heavy-tailed offspring laws. For the birth-death case,  the number of ancestral reproduction events grows asymptotically linearly with the length of the lineage.  By contrast for the heavy-tailed case,  the number of events grows exponentially with the length of the lineage. Thus the difference between the reproduction rates of ancestral vs typical cells is potentially unbounded. 
\item In Section 4,  we determine how the reproduction rate varies along the ancestral lineage. In particular for supercritical processes we see that the reproduction rate continuously increases along the lineage,  standing in contrast to the constant rate in the population. To explain the origin of ancestral variation in a homogeneous population, we also determine the ancestral reproduction rate as a function of the historical population size.
\item In Section 5, we discuss mutation rate variation on ancestral lineages seen in recent phylogenetic studies of the developing human embryo.
\end{itemize}

\color{black}

\section{The point process of ancestral reproduction} \label{sec:GW_T}
We consider a BGW branching process with initial population size one,  each cell independently reproducing at rate $r$ to be replaced by $k$ new cells with probability $p_k$. The offspring generating function is $f(s)=\sum_{k=0}^\infty p_ks^k$. The number of cells at time $t$ is $N_t$, with process generating function
\begin{align*}
F_t(s) = \mathbb{E}[ s^{N_t}].
\end{align*}
Generating functions are the key tool of branching process analysis thanks to the fact that the branching property (independence among cells) ensures the semigroup property $F_{t_1} \left( F_{t_2}(s) \right) = F_{t_1 + t_2 }(s)$.  Setting $u(s) := r( f(s) - s)$ the process generating function $F_t$ satisfies the Kolmogorov forward and backward equations
\begin{align} \label{eq:kolmogorov}
\frac{ \partial}{ \partial t} F_t(s) =  u(s) \frac{ \partial}{ \partial s} F_t(s) \qquad \text{and} \qquad \frac{ \partial}{ \partial t} F_t(s) =  u (F_t(s))
\end{align} 
with initial condition $F_0(s) = s$. See for instance Athreya and Ney \cite{AN}. In particular, $F_t(s)$ is determined by $r$ and $(p_k)_{k \geq 0}$, and in several cases it is possible to solve (\ref{eq:kolmogorov}) to compute $F_t(s)$ explicitly. 

Given that the population is alive at time $T$, sample a cell uniformly from the population at time $T$. 
Our first main result gives an explicit characterisation of the law of reproduction events along the uniform ancestral lineage.

\begin{thm}\label{main}Given that the population is alive at time $T$, sample a cell uniformly from the population at time $T$. On their ancestral lineage there are $j$ reproduction events of sizes $\ell_1,..,\ell_j$ during the time intervals $[t_1+\mathrm{d}t_1),..,[t_j,t_j+\mathrm{d}t_j)$ with probability 
\begin{align} \label{eq:ellrate}
\frac{e^{-rT}}{(1-F_T(0))}\int_0^1\left(\prod_{i=1}^jr\ell_ip_{\ell_i}F_{T-t_i}(s)^{\ell_i-1}\mathrm{d}t_i\right)ds.
\end{align}
\end{thm}

It is possible to give a more probabilistic statement of Theorem \ref{main}:

\begin{thm} \label{main2} Let $S$ be a random variable with density $F_T'(s)/(1-F_T(s))$ for $s\in[0,1]$. Then independently for each $\ell$, size $\ell$ reproduction events occur along the uniform ancestral lineage according to a time inhomogeneous Poisson point process with intensity function
\begin{align} \label{eq:prate}
r_\ell(S,t) := r \ell p_\ell F_{T-t}(S)^{\ell-1}.
\end{align}
\end{thm}

Properties of $F_t(s)$ in conjunction with \eqref{eq:ellrate} or \eqref{eq:prate} enable an immediate qualitative comment on the reproduction rates along the uniform ancestral lineage.
Write $m := \sum_{k \geq 0} kp_k$ for the mean of the offspring distribution. We say the tree is supercritical (respectively critical, subcritical) if $m >1$ (resp.\ $m=1$, $m<1$).
In the supercritical case with no deaths (i.e.\ $p_0=0$), the function $F_t(s)$ is monotone decreasing in the $t$ variable. (This can be seen from \eqref{eq:kolmogorov}.) As such, we then see from \eqref{eq:ellrate} that for each $\ell \geq 2$, the rate of reproduction events of size $\ell$ is increasing for $t \in [0,T]$ along the uniform ancestral lineage. Conversely, in the subcritical case, $F_t(s)$ is monotone increasing in the $t$ variable, and consequently, the rate of size $\ell \geq 2$ reproduction events is decreasing for $t \in [0,T]$ along the uniform ancestral lineage.

Further implications of Theorem \ref{main} will be seen in Sections \ref{countsec} and \ref{Ratesec}, where we shall determine the total number of reproduction events and the local rate of reproduction along the ancestral lineage, both of which are readily computable for birth-death and heavy-tailed branching processes. The remainder of the present section is dedicated to the proof of Theorem \ref{main}.

First we establish notation for the proof.  We follow the Ulam-Harris labelling system, in which each cell in the BGW process is associated with a label in $\mathcal{T}=\cup_{n=0}^\infty\mathbb{N}^n$. The first cell, born at time zero, is labelled by $\varnothing$, the empty word. When the initial cell $\varnothing$ dies and has $k \geq 0$ children, these children are labelled $(1),\ldots,(k)$. More generally, when a cell associated with a label $u = (u_1,\ldots,u_m)$ dies and has $k$ children, these children are labelled by the $k$ concatenations $(u1),\ldots,(uk)$. The set $\mathcal{T}$ is endowed with the partial ordering $\prec$ defined by $(u_1,..,u_m)\prec(v_1,..,v_n)$ if and only $m< n$ and $(u_1,..,u_m)=(v_1,..,v_m)$.  Write $\preceq$ for $\prec$ or $=$. In words, $u\prec v$ means that $u$ is ancestral to $v$. 

Let $(L_v : v \in \mathcal{T})$ be a collection of i.i.d.  non-negative integer-valued random values each distributed as $(p_k)$,  representing the children numbers of each cell, and let $(\tau_v:v\in\mathcal{T})$ be a collection of i.i.d. Exponential random variables with rate $r$, representing the cells' lifespans. Cell $v$ has birth time $\sigma_v=\sum_{u\prec v}\tau_u$ and reproduction time $\rho_v=\sum_{u\preceq v}\tau_u$. The cells alive at time $t\geq0$ are given by the set
\ba
\mathcal{N}_t=\left\{v=(v_1,..,v_n)\in\mathcal{T}:t\in[\sigma_v,\rho_v)\text{ and }v_i\leq L_{(v_1,..,v_{i-1})}\text{ for }i=1,..,n \right\},
\ea
and the number of cells at time $t$ is denoted $N_t=|\mathcal{N}_t|$.

For each $v\in\mathcal{T}$, we define a point process
$$
R_v=\sum_{u\prec v}\delta_{(L_u,\rho_u)}
$$
{characterising reproduction on the ancestral lineage of $v$. This point process is a measure on $\mathbb{N} \times [0,\rho_v]$.

On the event that $N_T$ is positive, sample uniformly a cell $V_T$ from $\mathcal{N}_T$. The purpose of Theorem \ref{main} is to give an explicit characterisation of the distribution of $R_{V_T}$.

In the spirit of branching process spine arguments, we now define a random sequence of elements running forwards through the population tree $\cup_{t\geq0}\mathcal{N}_t$.  First, on the event that the first cell $\varnothing$ has $k\geq1$ children, let $\mu_1$ be uniformly sampled from the set $\{1,..,k\}$ of $\varnothing$'s children. Continuing in this vein, on the event that the cell $(\mu_1,\ldots,\mu_n)$ has $k \geq 1$ children, we let $(\mu_1,..,\mu_n,\mu_{n+1})$ be uniformly sampled from the children $(\mu_1,\ldots,\mu_n,1),\ldots,(\mu_1,\ldots,\mu_n,k)$ of $(\mu_1,\ldots,\mu_n)$. The sequence $\{(\mu_1,..,\mu_n):n=0,1,2,..\}$ is the \textit{spine}, which can be mapped to an element $W_T$ of the population at time $T$. On the event that the set $\mathcal{W}_T=\mathcal{N}_T\cap\{(\mu_1,..,\mu_n):n=0,1,2,..\}$ is non-empty, let $W_T$ be the unique element of $\mathcal{W}_T$. Note however that if the spine follows a lineage which dies then $\mathcal{W}_T$ may be empty even if $\mathcal{N}_T$ is non-empty,  so $W_T$ is not always defined. We shall write $\{W_T\in\mathcal{N}_T\}$ for the event that $W_T$ is defined.

We have so far defined two random elements of $\mathcal{N}_T$. On the one hand we have $V_T$ chosen uniformly from the population $\mathcal{N}_T$ at time $T$.  On the other hand, we have the spine element $W_T$ defined by following uniformly chosen children through the tree forwards in time which leads to a non-uniform distribution on the elements of $\mathcal{N}_T$. Importantly, while reproduction along the ancestral lineage of $V_T$ may not reflect that of typical members of the population, reproduction along the ancestral lineage of $W_T$ coincides with the law of reproduction in the population. Our next result is the chief tool in our approach, characterising the relationship between reproduction on the ancestral lineages of $V_T$ and $W_T$.

\begin{lemma}\label{spineidentity}
For any measurable function $G$ from the space of point processes on $\mathbb{N}\times[0,T]$ to $[0,\infty)$,
\begin{align*}
\mathbb{E}\left[\mathrm{1}_{\{N_T \geq 1\}}G(R_{V_T})\right] =\mathbb{E}\left[\frac{\mathrm{1}_{\{W_T \in \mathcal{N}_T\}}}{N_{T}}\left(\prod_{v\prec W_T}L_v\right)G(R_{W_T}) \right].
\end{align*}

\end{lemma}
\begin{proof}Let $\mathcal{F}$ be the sigma-algebra generated by the random variables $(L_v,\tau_v:v\in\mathcal{T})$ defining the evolution of the population. First we look at the event $\{W_T=v\}$ conditional on $\mathcal{F}$. As the spine traverses from the initial individual $\varnothing$ through to $v$, at each reproduction event it must choose the `correct' child to follow in order for $\{W_T=v\}$ to hold, each choice being made correctly with probability given by one divided by the number of children. So 
$$\mathbb{P}[W_T=v|\mathcal{F}]=\frac{\mathrm{1}_{\{v \in \mathcal{N}_T,\prod_{u\prec v}L_u>0\}}}{\prod_{u\prec v}L_u}.$$
By rearranging, we obtain
$$
\mathrm{1}_{\{v \in \mathcal{N}_T\}}=\mathbb{P}\left[\mathrm{1}_{\{W_T=v\}}\prod_{u\prec v}L_u\Big|\mathcal{F}\right].
$$
It follows that for measurable functions $G$ on the space of point measures on $\mathbb{N} \times [0,T]$.
\begin{align*}
\frac{\mathrm{1}_{\{N_T \geq 1\}}}{N_{T}}\sum_{v\in\mathcal{N}_{T}}G(R_v)
&=\frac{\mathrm{1}_{\{N_T \geq 1\}}}{N_{T}}\sum_{v\in\mathcal{T}}\mathbb{P}\left[\mathrm{1}_{\{W_T=v\}}\prod_{u\prec v}L_u\Big|\mathcal{F}\right]G(R_v)\\
&=\mathbb{P}\left[\frac{\mathrm{1}_{\{N_T \geq 1\}}}{N_{T}}\sum_{v\in\mathcal{T}}\mathrm{1}_{\{W_T=v\}}\left(\prod_{u\prec W_T}L_u\right)G(R_{W_T})\Big|\mathcal{F}\right]\\
&=\mathbb{P}\left[\frac{\mathrm{1}_{\{W_T \in \mathcal{N}_T\}}}{N_{T}}\left(\prod_{u\prec W_T}L_u\right)G(R_{W_T})\Big|\mathcal{F}\right].
\end{align*}
Take expectations to obtain the result.
\end{proof}
Lemma \ref{spineidentity} says that the relationship between the ancestral lineages of $V_T$ and $W_T$ depends on the reciprocal of the population size at time $T$. The next result makes sense of the population size as the sum of subpopulations descending from elements of the spine. 
\begin{lemma}\label{summingspine}
We have
\begin{align*}
\mathbb{E}\left[\frac{\mathrm{1}_{\{W_T \in \mathcal{N}_T\}}}{N_{T}}\Big|\mathcal{G}\right]=\mathrm{1}_{\{W_T \in \mathcal{N}_T\}}\int_0^1\prod_{v\prec W_T}F_{T-\rho_v}(s)^{L_v-1}ds,
\end{align*}
where $\mathcal{G}$ is the sigma-algebra generated by $(L_v,\tau_v:v\prec W_T)$ (which carries information on the identity of the spine and reproduction along the spine).
\end{lemma}
\begin{proof}[Proof of Lemma \ref{summingspine}]
For any $w\in\mathcal{N}_T$, we can write $$N_T=1+\sum_{v\prec w}\sum_{u\in\text{sis}(v)}N_{T,u},$$
where sis$(v)$ is the set of sisters of $v$ and $N_{T,u}$ is the number of descendants of $u$ alive at time $T$, and hence
$$
\frac{\mathrm{1}_{\{W_T \in \mathcal{N}_T\}}}{N_{T}}=\frac{\mathrm{1}_{\{W_T \in \mathcal{N}_T\}}}{1+\sum_{v\prec W_T}\sum_{u\in\text{sis}(v)}N_{T,u}}.
$$
Then using the fact that $\int_0^1 s^k \mathrm{d}s = \frac{1}{k+1}$,
\begin{eqnarray}\frac{\mathrm{1}_{\{W_T \in \mathcal{N}_T\}}}{N_{T}}&=&\mathrm{1}_{\{W_T \in \mathcal{N}_T\}}\int_0^1s^{\sum_{v\prec W_T}\sum_{u\in\text{sis}(v)}N_{T,u}}ds\nonumber\\
&=&\mathrm{1}_{\{W_T \in \mathcal{N}_T\}}\int_0^1\prod_{v\prec W_T}\prod_{u\in\text{sis}(v)}s^{N_{T,u}}ds.\label{conde}
\end{eqnarray}
But for $v\prec W_T$ and $u\in\text{sis}(u)$,  the $N_{T,u}$ are conditionally independent given $\mathcal{G}$ and are distributed as $N_{T-\rho_{m(u)}}$, where $m(u)$ is $u$'s mother.  So taking the expectation of (\ref{conde}) conditional on $\mathcal{G}$ gives the result.
\end{proof}

Now we have the ingredients to complete the proof of Theorem \ref{main}.
\begin{proof}[Proof of Theorem \ref{main}]

By Lemma \ref{spineidentity},
\ba
\mathbb{E}\left[\mathrm{1}_{\{N_T \geq 1\}}G\left(R_{V_T}\right)\right]=\mathbb{E}\left[\frac{\mathrm{1}_{\{W_T \in \mathcal{N}_T\}}}{N_{T}}\left(\prod_{v\prec W_T}L_v\right)G(R_{W_T}) \right].\ea
Then using the tower rule along with Lemma \ref{summingspine} gives that
\begin{align} \label{eq:makela}
\mathbb{E}\left[\mathrm{1}_{\{N_T \geq 1\}}G\left(R_{V_T}\right)\right]=\mathbb{E}\left[\mathrm{1}_{\{W_T \in \mathcal{N}_T\}}G(R_{W_T})\int_0^1\prod_{v\prec W_T}L_vF_{T-\rho_v}(s)^{L_v-1}ds\right].
\end{align}
As mentioned above,  reproduction along the ancestral lineage of the spine by definition occurs at the same rate as the natural rate in the population. Thus we have the simple relation
\begin{align} \label{eq:makela2}
&\mathbb{P} \left( W_T \in \mathcal{N}_T, \text{The lineage of $W_T$ has $j$ reproductions of sizes $\ell_1,\ldots,\ell_j$ in $[t_1+\mathrm{d}t_1),\ldots,[t_j,t_j+\mathrm{d}t_j)$} \right) \nonumber \\
&= r^j e^{-rT} \prod_{i=1}^j p_{\ell_i} \mathrm{d}t_i.
\end{align}
Plugging \eqref{eq:makela2} into the right-hand-side of \eqref{eq:makela}, we obtain
\begin{align*} 
\mathbb{E}\left[\mathrm{1}_{\{N_T \geq 1\}}G\left(R_{V_T}\right)\right]= r^j e^{-rT} \prod_{i=1}^j p_{\ell_i} \mathrm{d}t_i \int_0^1 \prod_{i=1}^j \ell_i F_{T-t_i}(s)^{\ell_i-1} \mathrm{d}s,
\end{align*}
thereby completing the proof of Theorem \ref{main}. 
\end{proof}

We now turn to the proof of Theorem \ref{main2}. 
We begin with the following lemma.

\begin{lemma} \label{lem:totalrate}
We have
\begin{align} \label{eq:tchaik8}
r\int_0^T f'(F_{T-t}(s)) \mathrm{d}t &= rT + \log F_T'(s) .
\end{align}
\end{lemma}
\begin{proof}
Towards calculating \eqref{eq:tchaik5}, recall from \eqref{eq:kolmogorov} that setting $u(s) = r(f(s)-s)$ we have $\partial/\partial t F_t(s) = u(F_t(s))$. It follows that
\begin{align} \label{eq:dde}
\frac{\partial}{\partial s} \frac{ \partial }{ \partial t} F_t(s) = u'(F_t(s)) \frac{\partial}{\partial s} F_t(s),
\end{align}
and then
\begin{align*} 
r\int_0^T f'(F_{T-t}(s)) \mathrm{d}t &= r\int_0^T f'(F_{t}(s)) \mathrm{d}t  \nonumber\quad &\text{(changing variable $t \mapsto T-t$)}    \\
&= rT + \int_0^T u'(F_{T-t}(s)) \mathrm{d}t   \nonumber  \quad&\text{(using $rf'(s) = r + u'(s)$)} \\
&= rT + \int_0^T \frac{ \frac{\partial}{\partial s} \frac{ \partial }{ \partial t} F_t(s)  }{ \frac{\partial}{\partial s} F_t(s) }  \mathrm{d}t   \nonumber  \quad &\text{(applying \eqref{eq:dde})}\\
&= rT + \int_0^T \frac{\partial}{\partial t} \log \left( \frac{\partial}{\partial s}F_t(s) \right) \mathrm{d}t.
\end{align*}
Performing the integral and using that $F_0(s) = s$ and hence $F_0'(s) = 1$, we obtain
\begin{align*}
r\int_0^T f'(F_{T-t}(s)) \mathrm{d}t &= rT + \log F_T'(s) .
\end{align*}
\end{proof}

We are now equipped to prove Theorem \ref{main2}.

\begin{proof}[Proof of Theorem \ref{main2}]
We verify that the Cox process in question has the same event probabilities as those given in Theorem \ref{main}.

We begin by noting that since $F_T(1) = 1$, that $F_T'(s)/(1-F_T(0))$ clearly integrates to $1$ for $s \in [0,1]$. 

Consider now the following general fact. If we a time inhomogenous Poisson process with intensity $\lambda:[0,T] \to [0,\infty)$, then the probability this process has its events occuring in $\mathrm{d}t_1,\ldots,\mathrm{d}t_j$ is given by 
\begin{align*}
\exp \left\{ - \int_0^T \lambda(t) \mathrm{d}t \right\} \prod_{i=1}^j \lambda(t_i) \mathrm{d}t_i.
\end{align*}
Consequently, if conditionally on $S$ and independently for each $\ell$, size $\ell$ reproduction events occur along the uniform ancestral linage according to a Poisson process of rate $r_\ell(S,t)$ given in \eqref{eq:prate}, then the conditional probability of $j$ total reproduction events of sizes $\ell_1,\ldots,\ell_j$ in $\mathrm{d}t_1, \ldots,\mathrm{d}t_j$ is given by 
\begin{align*}
\exp\left\{ - \sum_{ \ell \geq 1} \int_0^T r \ell p_\ell F_{T-t}(S)^{\ell-1} \mathrm{d}t \right\} \prod_{i=1}^j r \ell_i p_{\ell_i} F_{T-t_i}(S)^{\ell_i-1}\mathrm{d}t_i.
\end{align*}
Now interchanging the order of summation and integration and thereafter using the fact that $f'(s) = \sum_{\ell \geq 1} \ell p_\ell s^{\ell-1}$ to obtain the first equality below, and Lemma \ref{lem:totalrate} to obtain the second, we have
\begin{align*}
\exp\left\{ - \sum_{ \ell \geq 1} \int_0^T r \ell p_\ell F_{T-t}(S)^{\ell-1} \mathrm{d}t \right\} = \exp\left\{ -r \int_0^T f'\left( F_{T-t}(S) \right) \mathrm{d}t \right\} = e^{-rT}/F_T'(S).
\end{align*}
It follows that the probabilities associated with the Cox process described are given by 
\begin{align*}
\int_0^1 \frac{F_T'(s)}{1-F_T(0)} ~ \frac{e^{-RT}}{F_T'(s)} \prod_{i=1}^j r \ell_i p_{\ell_i} F_{T-t_i}(S)^{\ell_i-1}\mathrm{d}t_i \mathrm{d}s,
\end{align*}
which agrees exactly with the integral formula in Theorem \ref{main}.
\end{proof}

\section{The total number of ancestral reproduction events}\label{countsec}
\subsection{The law of the number of events}
Having determined the full distribution of reproduction events on the ancestral lineage of the cell sampled at time $T$,  we now turn our attention to the total number of reproduction events along the lineage, which we denote by $E_T$. The number $E_T$ can also be thought of as the generation number of the sampled cell.

We note that the generation number of individuals in a branching process can be understood more generally in terms of branching random walks.  A branching random walk is a branching process where each individual has a spatial location chosen according to its parent's location plus some random jump.  A special case is that the population begins with one individual at position zero on the real line, and that every other individual is located one integer to the right of their parent - so an individual's position is exactly their generation number.  Since decades ago, there are results on the large-time spatial distribution of branching random walks \cite{AN}, and thus much is known about the large-time behaviour of the random generation number $E_T$. 

Samuels \cite{samuels} provides an explicit analysis of generation numbers, showing that for an offspring distribution with mean $m=\sum_{k=1}^\infty kp_k$ taking values in $(1,\infty)$,  $E_T$ has asymptotically Gaussian behaviour with expectation $mrT$ and variance $mrT$:
\begin{align} \label{eq:samuels}
\lim_{T \to \infty} \mathbb{P} \left( \frac{E_T - mrT}{ \sqrt{ mrT} } < x \right) = \int_{-\infty}^x \frac{ e^{ - u^2/2}}{ \sqrt{2 \pi }} \mathrm{d}u.
\end{align}
In fact, Samuels' limit theorem holds in the more general setting of age-dependent (and hence non-Markovian) branching processes.  

In the Markovian setting, our next result an explicit formula for the distribution of $E_T$ at finite times $T$.

\begin{proposition}\label{thm:total}
We have 
\begin{align*}
\mathbb{P}( E_T = j | N_T \geq 1) = \frac{e^{-rT}}{\mathbb{P}(N_T \geq 1) } \int_0^1 \frac{ (rT + \log F_T'(s) )^j}{j!} \mathrm{d}s.
\end{align*}
\end{proposition}

\begin{proof}[Proof of Proposition \ref{thm:total}]
This may be proved using the representation of the events along the uniform ancestral lineage as a Cox process, but here we give a direct proof using Theorem \ref{main}.

Beginning with the statement of Theorem \ref{main} and then summing over $\ell_i$ and integrating over $t_i$, 
\begin{align} \label{eq:tchaik6}
\mathbb{P}(E_T = j , N_T \geq 1) = \sum_{\ell_1,\ldots,\ell_j \geq 1 } \int_{0 < t_1 < \ldots < t_j < T} r^j e^{- rT} \int_0^1 \prod_{i=1}^j \ell_i p_{\ell_i}   F_{T-t_i}(s)^{\ell_i-1} \mathrm{d}s ~  \mathrm{d}t_1 \ldots \mathrm{d}t_j.
\end{align}
The sum over $\ell_i$ is straightforward. Since $f'(s) = \sum_{\ell \geq 1} \ell p_\ell s^{\ell-1}$, \eqref{eq:tchaik6} reduces to
\begin{align*} 
\mathbb{P}(E_T = j , N_T \geq 1) =  r^j e^{- rT}   \int_{0 < t_1 < \ldots < t_j < T}\int_0^1 \prod_{i=1}^j f'\left(F_{T-t_i}(s) \right) \mathrm{d}s ~  \mathrm{d}t_1 \ldots \mathrm{d}t_j.
\end{align*}
Using the symmetry of the integral in $t_1,\ldots,t_j$, and changing the order of integration, this reduces further to
\begin{align} \label{eq:tchaik5}
\mathbb{P}(E_T = j , N_T \geq 1) =\frac{r^j e^{- rT}}{j!} \int_0^1 \prod_{i=1}^j \left\{  \int_0^T f'\left( F_{T-t_i}(s) \right) \mathrm{d}t_i  \right\} \mathrm{d}s.
\end{align}
Applying \eqref{eq:tchaik8} to \eqref{eq:tchaik5} we obtain 
\begin{align*}
\mathbb{P}(E_T = j , N_T \geq 1) =\frac{e^{- rT}}{j!} \int_0^1 (rT + \log F_T'(s))^j \mathrm{d}s,
\end{align*}
which, after dividing through by $\mathbb{P}(N_T \geq 1)$ gives Proposition \ref{thm:total}.
\end{proof}

Observe that the statement of Proposition \ref{thm:total} can be rewritten as 
\begin{align} \label{eq:total2}
\mathbb{P}( E_T = j | N_T \geq 1) = \int_0^1 \left\{ e^{ - (rT + \log F_T'(s) )}  \frac{ (rT + \log F_T'(s) )^j}{j!} \right\}  \frac{F_T'(s)}{1 - F_T(0)}  \mathrm{d}s.
\end{align}
The integrand of (\ref{eq:total2}) is the product two terms: one,  the probability mass function of Poisson distribution with mean $rT + \log F_T'(s)$, and two, the function $F_T'(s)/(1-F_T(0))$ which is a probability density on $[0,1]$. The equation \eqref{eq:total2} therefore states that $E_T$ is a mixture of Poisson distributions, that is, $E_T$ has the law of a Poisson random variable with random mean $rT + \log F_T'(S)$, where $S$ is a $[0,1]$-valued random variable distributed according to the density $F_T'(s)/(1-F_T(0))$. 

\subsection{The expected number of events} \label{sec:gap}
As a result of the probabilistic representation \eqref{eq:total2}, we see that the expectation of $E_T$ coincides with the expectation of the random variable $rT + \log F_T'(S)$, that is 
\begin{align} \label{eq:expect}
\mathbb{E}[E_T | N_T \geq 1 ] = \int_0^1 (r T + \log F_T'(s)) \frac{F_T'(s)}{1 - F_T(0)} \mathrm{d}s.
\end{align}
Since $F_T'(s)$ is an increasing function of $s$ and $F_T'(1)=\mathbb{E}[N_T ] = e^{r(m-1)T}$,  (\ref{eq:expect}) is bounded above by
\begin{align*} 
\int_0^1 (r T + \log F_T'(1)) \frac{F_T'(s)}{1 - F_T(0)} \mathrm{d}s= \int_0^1 rmT\frac{F_T'(s)}{1 - F_T(0)} \mathrm{d}s=rmT.
\end{align*}
That is, $\mathbb{E}[E_T | N_T \geq 1 ]\leq rmT$, and according to Samuel's large time limit result (\ref{eq:samuels}), this upper bound is asymptotically attained.  In numerical calculations however we observed unfortunately slow convergence. For the remainder of the section we study the (non-negative) gap.
\begin{align} \label{eq:gap}
\mathrm{Gap}_T := rmT - \mathbb{E}[E_T | N_T \geq 1 ] 
\end{align}
between $rmT$ and the expectation $\mathbb{E}[E_T | N_T \geq 1]$. 
\begin{lemma} \label{lem:Gap}
Under mild conditions, $\mathrm{Gap}_T$ converges to a finite non-negative limit $\mathrm{Gap}_\infty$.
\end{lemma}
As an aside, Lemma \ref{lem:Gap} is closely related to a result of the first author and Shneer \cite{Seva} regarding the empirical mean of generations in the population at time $T$, which we shall denote $\mathcal{E}_T$.  Their result says that $rmT-\mathcal{E}_T$  converges almost surely as $T\rightarrow\infty$ to a finite random random variable,  resembling Lemma \ref{lem:Gap} because $\mathbb{E}[\mathcal{E}_T]=\mathbb{E}[E_T]$. 
\begin{proof}[Sketch proof of Lemma \ref{lem:Gap}]
In this sketch proof, we will assume familiarity with the limit theory of BGW processes.  We begin with the equality
\begin{align} \label{eq:gap2}
\mathrm{Gap}_T := \int_0^1 (r(m-1)T - \log F_T'(s)) \frac{F_T'(s)}{1 - F_T(0)} \mathrm{d}s,
\end{align}
which follows from \eqref{eq:gap}, \eqref{eq:expect} and the fact that $F_T'(s)/(1-F_T(0))$ is a probability measure on $[0,1]$. 

Now we separately analyse the supercritical, critical, and subcritical cases.  For the supercritical case $m>1$, consider the unit-mean martingale $Z_t := N_te^{ - r(m-1)t}$, and define $\varphi_t(\theta) := \mathbb{E}[Z_t e^{ - \theta Z_t}]$. Changing variable $s \mapsto \theta e^{- \theta e^{ - r(m-1)T}}$ in \eqref{eq:gap2}, we see that
\begin{align} \label{eq:supergap}
\mathrm{Gap}_T = - \int_0^\infty (\theta e^{ - r(m-1)T} + \log \varphi_T(\theta)) \frac{ \varphi_T(\theta)}{ 1 - F_T(0) } \mathrm{d}\theta,
\end{align}
where again, it may be shown that $\varphi_T(\theta)/(1 - F_T(0))$ is a probability density on $(0,\infty)$. Under the Kesten-Stigum condition \cite{AN} that $m \in (1,\infty)$ and $\sum_{ k =1}^\infty p_kk \log k  < \infty$, the function $\varphi_T(\theta)$ converges to a limit $\varphi(\theta)$, and hence \eqref{eq:supergap} converges to
\begin{align*}
\mathrm{Gap}_\infty := - \int_0^\infty \log \varphi(\theta) \frac{ \varphi(\theta) }{1 - F_\infty(0) } \mathrm{d}\theta
\end{align*}
as $T \to \infty$. We note that $1 - F_\infty(0)$ is the survival probability.

In the critical case $m=1$ with $\sum_{k=1}^\infty p_kk^2<\infty$, Yaglom's theorem \cite{AN} states that  $\lim_{T \to \infty} F_T'(e^{ - 2 \theta/\sigma^2 rT} ) = (1 + \theta)^2$, and $\lim_{ T \to \infty} (1 - F_T(0)) = \frac{2}{\sigma^2 rT}$. Plugging these facts into \eqref{eq:gap2} using the change of variable $s = e^{ - \frac{2}{\sigma^2}{rT} \theta}$ to obtain the first equality below, and using the change of variable $\log(1+\theta) = \zeta$ and the gamma integral to obtain the second, we have 
\begin{align*}
\mathrm{Gap}_\infty = 2\int_0^\infty \log(1 + \theta) \frac{\mathrm{d}\theta}{(1+\theta)^2} =2.
\end{align*}

Finally in the subcritical case $m < 1$,  $ F_T'(s) e^{ -r(m-1)T } $ converges to a function $B(s)$ as $T\rightarrow\infty$ \cite{AN}, and moreover $\lim_{T \to \infty} e^{ - r(m-1)T} (1 - F_T(0)) = c_{\mathrm{surv}}$ exists. Thus \eqref{eq:gap2} converges to
\begin{align*}
\mathrm{Gap}_\infty = -\frac{1}{c_{\mathrm{surv}}} \int_0^1 \log (B(s)) B(s) \mathrm{d}s.
\end{align*}
While in each case the non-negativity of $\mathrm{Gap}_\infty$ may not be transparent from the expressions for $\mathrm{Gap}_\infty$, it follows immediately from the fact that each $\mathrm{Gap}_T$ is non-negative.
\end{proof}

We emphasise that the constant $\mathrm{Gap}_\infty$ in Lemma \ref{lem:Gap} is universal in all critical processes with finite variance.  Namely,  sampling from a critical process at a large time $T$, 
\begin{align*}
\mathbb{E}[E_T |N_T \geq 1 ] = rT - 2 + o(1).
\end{align*}
We apply the correction of 2 to the mean of $E_T$ in Samuel's central limit theorem (\ref{eq:samuels}) for a critical birth-death process,  
plotted in Figure \ref{CLTfig}. 
The exact probability distribution of $E_T$ is also plotted, whose computation is described next.

\subsection{The birth-death process}
For an explicit computation we now specialise Proposition \ref{thm:total} to the birth-death process,  defined by the binary offspring distribution $p_k=0$ for $k\not\in\{0,2\}$. Let $\beta$ and $\alpha$ denote the birth and death rates respectively, so $r = \alpha+\beta$, $p_0 = \alpha/(\alpha+\beta)$ and $p_2 = \beta/(\alpha+ \beta)$. From Kolmogorov's equations \eqref{eq:kolmogorov} it is possible to compute the generating function $F_t(s)$ for the birth-death process. Indeed, using the offspring generating function $f(s) = \frac{ \alpha}{\alpha+\beta} + \frac{ \beta}{ \alpha + \beta } s^2$ we have 
\begin{align} \label{eq:bd}
F_t(s) = \frac{ \alpha(1 - s) e^{ (\beta-\alpha)t} + \beta s - \alpha }{ \beta(1 - s) e^{ (\beta - \alpha) t } + \beta s - \alpha } ,
\end{align} seen for example in \cite{AN}.  The generating function (\ref{eq:bd}) holds for the non-critical case $\alpha\not=\beta$. We shall come back to the critical case later. Differentiating (\ref{eq:bd}),
\begin{align} \label{eq:bd der}
F_T'(s)  = \frac{  (\beta-\alpha)^2 e^{(\beta-\alpha)T} }{ \left( \beta(1 - s) e^{ (\beta-\alpha)T} + \beta s - \alpha \right)^2 } . 
\end{align}
Plugging \eqref{eq:bd der} into Theorem \ref{thm:total} we have
\begin{align*}
\mathbb{P}(E_T = j , N_T \geq 1 ) = \frac{ (-1)^j 2^j  e^{ - (\alpha+\beta)T} }{ j!} \int_0^1 \left\{  \log\left( \frac{ \beta(1 - s) e^{ (\beta-\alpha)T} + \beta s - \alpha }{ (\beta-\alpha)e^{\beta T}} \right) \right\}^j \mathrm{d}s.
\end{align*}
Changing variable $x = \frac{ \beta(1 - s) e^{ (\beta-\alpha)T} + \beta s - \alpha }{ (\beta-\alpha)e^{2 \beta T}}$  we have 
\begin{align*}\mathbb{P}(E_T = j , N_T \geq 1 )= \frac{ (-1)^{j} 2^j e^{ - (\alpha+ \beta)T} }{j!} \frac{ (\beta - \alpha) e^{ \beta T } }{ \beta( e^{ (\beta-\alpha)T} - 1 ) }  \int^{ \frac{ \beta e^{(\beta-\alpha)T} - \alpha}{ (\beta- \alpha)e^{\beta T} } }_{ e^{ - \beta T } } \log(x)^j \mathrm{d}x.
\end{align*}
It is possible to perform the integral. We begin with the identity
\begin{align*}
\int_a^b \log( x)^j \mathrm{d}x = \sum_{ i = 0}^j \frac{j!}{i!} (-1)^{j-i} \left( b \log(  b)^i - a \log(a)^i \right),
\end{align*} 
from which we obtain
\begin{align*}
&\mathbb{P} \left( E_T=j,N_T \geq 1 \right)\\
&= \frac{ (-1)^{j} 2^j e^{ - (\alpha+ \beta)T} }{j!} \frac{ (\beta - \alpha) e^{ \beta T } }{ \beta( e^{ (\beta-\alpha)T} - 1 ) } \sum_{ i = 0}^j \frac{j!}{i!} (-1)^{j-i} \left( \frac{ \beta e^{(\beta-\alpha)T} - \alpha}{ (\beta- \alpha)e^{\beta T} }  \log\left( \frac{ \beta e^{(\beta-\alpha)T} - \alpha}{ (\beta- \alpha)e^{\beta T} } \right)^i -  e^{ - \beta T }  \log( e^{ - \beta T } )^i \right)\\
&= 2^j e^{ - \alpha T } \frac{ \beta - \alpha}{ \beta( e^{ (\beta-\alpha)T} - 1 ) } \sum_{ i = 0}^j \frac{1}{i!} \left( \frac{ \beta e^{(\beta-\alpha)T} - \alpha}{ (\beta- \alpha)e^{\beta T} }  (-1)^i \log\left( \frac{ \beta e^{(\beta-\alpha)T} - \alpha}{ (\beta- \alpha)e^{\beta T} } \right)^i -  e^{ - \beta T }  (\beta T)^i \right).
\end{align*}
Dividing through by $\mathbb{P}(N_T \geq 1)= \frac{ \beta - \alpha}{ \beta - \alpha e^{ - (\beta-\alpha)T} }$ , 
\begin{align}\label{ncc}
\mathbb{P} \left( E_T=j|N_T \geq 1 \right)= 2^j e^{ - \alpha T } \frac{ \beta - \alpha e^{ - ( \beta - \alpha )T } }{ \beta( e^{ (\beta-\alpha)T} - 1 ) } \sum_{ i = 0}^j \frac{1}{i!} \left( \frac{ \beta e^{(\beta-\alpha)T} - \alpha}{ (\beta- \alpha)e^{\beta T} }   \log\left( \frac{ (\beta- \alpha)e^{\beta T} }{ \beta e^{(\beta-\alpha)T} - \alpha} \right)^i -  e^{ - \beta T }  (\beta T)^i \right).
\end{align}
In the special case of the pure-birth process $\alpha=0$, otherwise known as the Yule process, the number of reproduction events on the ancestral lineage given by (\ref{ncc}) simplifies to
\begin{align}\label{Yule}
\mathbb{P} \left( E_T=j|N_T \geq 1 \right)= 2^j    \frac{  1  }{  e^{ \beta T} - 1} \sum_{ i = 0}^j \left( \mathrm{1}_{i=0} - e^{ - \beta T } \frac{ (\beta T)^i}{i!} \right) =      \frac{  2^j  }{  e^{ \beta T} - 1} \left\{ 1 - e^{ - \beta T } \sum_{ i = 0}^j  \frac{ (\beta T)^i}{i!} \right\}.
\end{align}

The derivation to arrive at (\ref{ncc}) was based on the non-critical case $\alpha \neq \beta$.  As for the critical case $\alpha = \beta$,  the fact that the probabilities in question are continuous functions of the parameters can be used: taking the limit $\alpha\rightarrow\beta$ in (\ref{ncc}),
\begin{align}\label{cc}
\mathbb{P} \left( E_T=j|N_T \geq 1 \right)= 2^j e^{ - \beta T } \frac{\beta T+1}{\beta T}  \sum_{ i = 0}^j \frac{1}{i!} \left( \frac{1+\beta T}{e^{\beta T}} \log \left( \frac{e^{\beta T} }{1+ \beta T }\right)^i - \frac{ (\beta T)^i}{ e^{\beta T} }  \right).
\end{align}

A comparison of the exact  (\ref{cc}) and approximate (\ref{eq:samuels}) distributions of $E_T$ is given in Figure \ref{CLTfig}. 

\begin{figure}\includegraphics[scale=0.1]{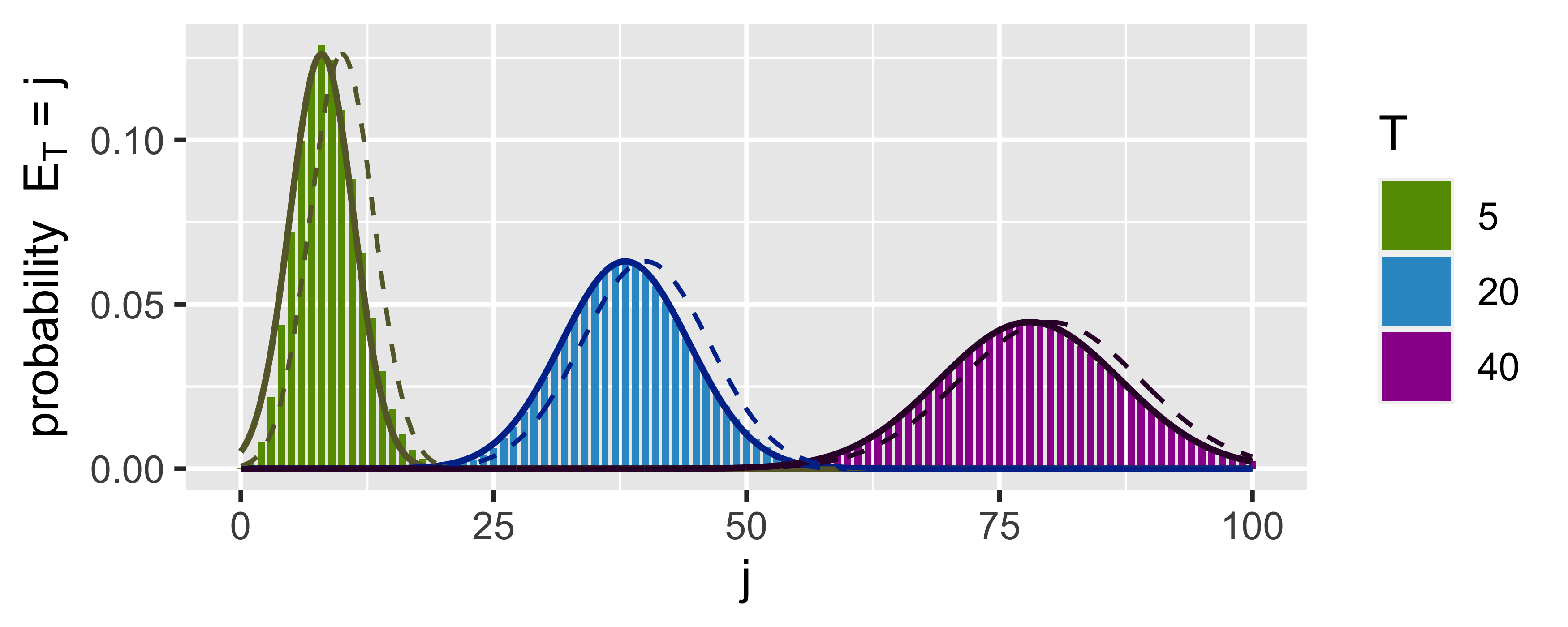}
\caption{A critical birth-death branching process with $p_0=p_2=1/2$ and $r=1$ is considered.  The bars represent the exact probabilities that $E_T=j$ given by (\ref{cc}),  the dashed lines show Samuels' Gaussian approximation (\ref{eq:samuels}),
and the solid lines show Samuels' approximation but with the mean corrected by Lemma \ref{lem:Gap}.}\label{CLTfig}
\end{figure} 

\subsection{The heavy-tailed branching process} \label{sec:heavy}
Another special case is the heavy-tailed branching process defined by reproduction rate $r =1$ and offspring probabilities $p_k=\frac{1}{k(k-1)}$ for $k\geq2$. The offspring generating function is given by 
\begin{align*}
f(s) = \sum_{ k \geq 2} \frac{1}{k(k-1)}s^k.
\end{align*}
It is easily verified that $f'(s) = - \log(1-s)$, and then integrating and using $f(0)=0,f(1)=1$ we find that $f(s) = s + (1-s)\log(1-s)$. It is straightforward to solve Kolmogorov's equation \eqref{eq:kolmogorov} to find that 
\begin{align*}
F_t(s) = 1 - (1-s)^{e^{-t}}.
\end{align*}
Consequently,
\begin{align*}
T + \log F_T'(s) = - (1 -e^{-T}) \log (1 - s).
\end{align*}
Finally, by Theorem \ref{thm:total} and the identity $\int_0^1 \log(x)^j \mathrm{d}x = j!$ we have
\begin{align*}
\mathbb{P}(E_T = j ) = \frac{e^{-T}}{j!} (1-e^{-T})^j.
\end{align*}
That is, the number of reproduction events on the uniformly sampled cell from the population at time $T$ has Geometric distribution with mean $e^T$.  The central limit theorem (\ref{eq:samuels}) proven in \cite{samuels} is not relevant here, because the  mean offspring number is infinite, and notably, the rate of ancestral reproduction per unit time grows to infinity with $T$ even though each cell in the population reproduces at constant rate one. This rather extreme example emphasises a basic conceptual point of the paper, that cells on an ancestral lineage may behave, statistically speaking, very differently from typical members of the population.

While this section discussed the total number of events on the sampled ancestral lineage,  which in a sense is a global view of the ancestral reproduction rate, in the next section we take a local view,  looking at how the reproduction rate varies along the ancestral lineage.
\section{Rate of reproduction along the ancestral lineage}\label{Ratesec}
\subsection{As a function of time}
We aim to calculate the reproduction rate at a specific timepoint on the ancestral lineage. Recall from Theorem \ref{main} that the probability that there are $j$ reproduction events on the ancestral lineage of sizes $\ell_1,..,\ell_{j}$ during time intervals $[t_1,t_1+\mathrm{d}t_1),..,[t_j,t_j+\mathrm{d}t_j)$ is
$$
\frac{e^{-rT}r^j}{(1-F_T(0))j!}\int_0^1\left(\prod_{i=1}^{j}\ell_ip_{\ell_i}F_{T-t_i}(s)^{\ell_i-1}\right)F_{T-t}(s)\mathrm{d}s\mathrm{d}t_1..\mathrm{d}t_{j}.
$$
Integrating over the reproduction times and summing over the reproduction sizes of $j-1$ of these events, the probability that there are $j$ reproduction events on the ancestral lineage, one of which has size $\ell$ and occurs during the time interval $[t,t+dt)$, is
\begin{align}
&\frac{e^{-rT}r^j}{(1-F_T(0))(j-1)!}\sum_{\ell_1,..,\ell_{j-1}\geq1}\int_{[0,T]^{j-1}}\int_0^1\left(\prod_{i=1}^{j-1}\ell_ip_{\ell_i}F_{T-t_i}(s)^{\ell_i-1}\right)\ell p_{\ell}F_{T-t}(s)^{\ell-1}\mathrm{d}s\mathrm{d}t_1..\mathrm{d}t_{j-1}dt\nonumber\\
&=\frac{e^{-rT}r^j}{(1-F_T(0))(j-1)!}\int_0^1\left(\prod_{i=1}^{j-1}\int_0^Tf'(F_{T-t_i}(s))\mathrm{d}t_i\right)\ell p_{\ell}F_{T-t}(s)^{\ell-1}\mathrm{d}s\mathrm{d}t\nonumber\\
&=\frac{e^{-rT}r\ell p_{\ell}}{(1-F_T(0))(j-1)!}\int_0^1\left(rT + \log F_T'(s)\right)^{j-1}F_{T-t}(s)^{\ell-1}\mathrm{d}s\mathrm{d}t\nonumber
\end{align}
where the last equality comes from (\ref{eq:tchaik8}). Summing over the number of reproduction events,  the probability that there is a reproduction event of size $\ell$ during the time interval $[t,t+dt)$ on the ancestral lineage is
\begin{align*}
&\sum_{j\geq1}\frac{e^{-rT}r\ell p_{\ell}}{(1-F_T(0))(j-1)!}\int_0^1\left(rT + \log F_T'(s)\right)^{j-1}F_{T-t}(s)^{\ell-1}\mathrm{d}s\mathrm{d}t\\
&=\frac{r\ell p_{\ell}}{1-F_T(0)}\int_0^1 F_T'(s)F_{T-t}(s)^{\ell-1}\mathrm{d}s\mathrm{d}t\\
&=\frac{r\ell p_{\ell}}{1-F_T(0)}\int_0^1 F_t'(F_{T-t}(s))F_{T-t}'(s)F_{T-t}(s)^{\ell-1}\mathrm{d}s\mathrm{d}t\quad&\text{(semigroup property)}\\
&=\frac{r\ell p_{\ell}}{1-F_T(0)}\int_{F_{T-t}(0)}^1 F_t'(z)z^{\ell-1}\mathrm{d}z\mathrm{d}t\quad&\text{(change of variable $z=F_{T-t}(s)$).}
\end{align*}
Summing over the possible sizes $\ell$ of the reproduction event, the probability that a reproduction event of any size occurs on the ancestral lineage during the time interval $[t,t+dt)$ is
\begin{align*}&\sum_{\ell\geq1}\frac{r\ell p_{\ell}}{1-F_T(0)}\int_{F_{T-t}(0)}^1 F_t'(z)z^{\ell-1}\mathrm{d}z\mathrm{d}t\\
&=\frac{r}{1-F_T(0)}\int_{F_{T-t}(0)}^1 F_t'(z)f'(z)\mathrm{d}z\mathrm{d}t,
\end{align*}
which gives the following.
\begin{proposition}\label{localrateprop}On the event that the population survives to time $T$, sample a cell uniformly at random from the population at this time. On the sampled cell's ancestral lineage at time $t<T$, reproduction happens at rate
\begin{align} \label{eq:rateeq}
R_T(t) := \frac{r}{1-F_T(0)}\int_{F_{T-t}(0)}^1 F_t'(s)f'(s)\mathrm{d}s.
\end{align}
\end{proposition}
In the special case of no deaths, that is $p_0=0$, we have that $F_t(0)=\mathbb{P}[N_t=0]=0$ for all $t$ and hence (\ref{eq:rateeq}) simplifies to
\begin{align} \label{eq:rateeqnodeath}
R(t) =r\int_{0}^1 F_t'(s)f'(s)\mathrm{d}s,
\end{align}which does not depend on the sampling time $T \geq t$.  In the absence of death we shall write $R(t)=R_T(t)$ for short. Perhaps the irrelevance of the sampling time is unsurprising. After all, the ancestor at time $T'<T$ of a uniformly chosen cell at time $T$ are themselves uniformly chosen from the population at time $T'$,  and thus the rate at which the ancestor at time $t<T'$ reproduces should not depend on whether the population was sampled at time $T'$ or $T$.  But this logic does not hold if the death rate is positive, because conditioning that the population is alive at the sampling time influences the historical population's rate of reproduction. 

We now use Proposition \ref{localrateprop} to compute the reproduction rate $R(t)$ explicitly in the case that the BGW process is a Yule tree.  Applying $f(s) = s^2$ and \eqref{eq:bd der} with $\alpha=0$ to \eqref{eq:rateeqnodeath},
\begin{align} \label{Yulerate}
R(t) = 2 \int_0^1 \frac{e^t}{(e^t(1-s)+s)^2 }s~ \mathrm{d}s = \frac{2}{(1-e^{-t})^2} ( 1 - (t+1)e^{-t}),
\end{align}
which is plotted in Figure \ref{Yuleratefig}. It is also possible to calculate $R(t)$ via an alternative route. The total number of events $E_T$ on the ancestral lineage up to time $T$, defined in Section 3, is related to $R(t)$ by $\mathbb{E}[E_T] = \int_0^T R(t) \mathrm{d}t,$ or equivalently $R(T) = \frac{\mathrm{d}}{\mathrm{d}T} \mathbb{E}[E_T]$. So one may compute the expectation of $E_T$ whose distribution is given in \eqref{Yule} for the Yule process, and thereafter differentiate with respect to $T$. We leave the details to the interested reader.

\begin{figure}\includegraphics[scale=0.1]{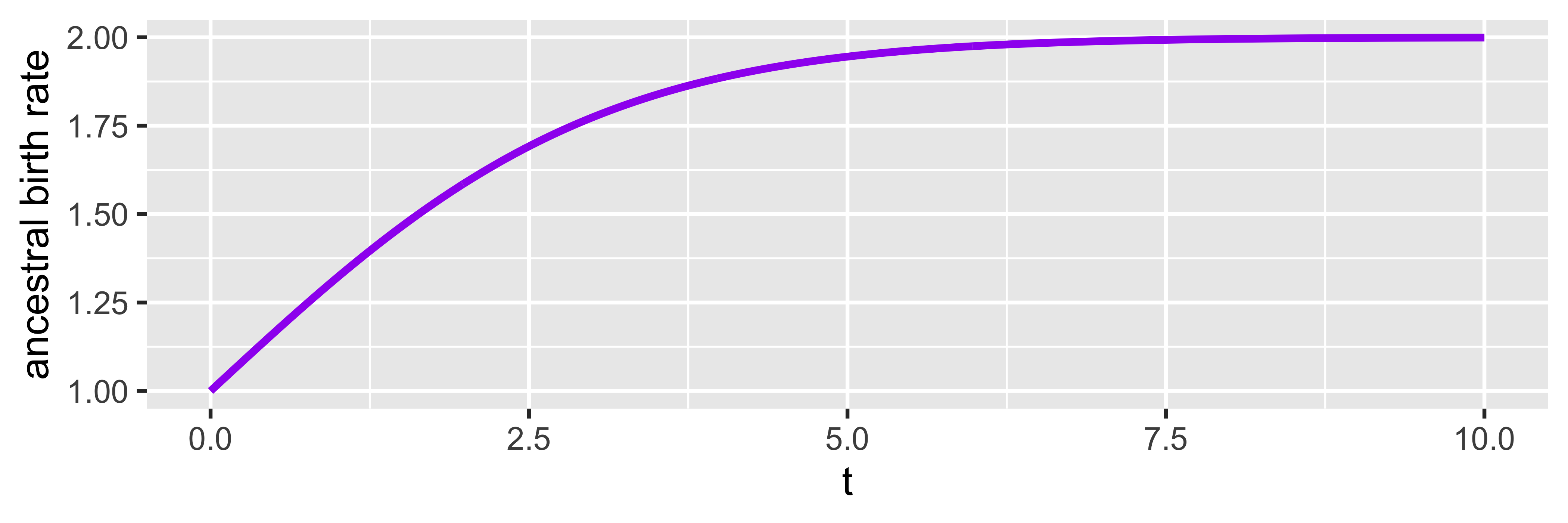}
\caption{For a Yule process with births at rate $1$, the birth-rate (\ref{Yulerate}) along the ancestral lineage of a uniformly sampled cell is plotted.}\label{Yuleratefig}
\end{figure} 

We turn to computing the rate function $R(t)$ associated with the heavy-tailed BGW process introduced in Section \ref{sec:heavy}, defined by $r=1$ and $p_k=\frac{1}{k(k-1)}$ for $k\geq2$. We have that
 $f'(s)=-\log(1-s)$ and $F'_t(s)=e^{-t}(1-s)^{e^{-t}-1}$, so the rate of reproduction on the ancestral lineage at time $t$ is
$$
R(t) = -r\int_0^1e^{-t}(1-s)^{e^{-t}-1}\log(1-s)ds=re^t.
$$
Thus despite a constant reproduction law among cells in the population, the reproduction rate increases exponentially with time along the ancestral lineage. 

The monotonicity of $R_T(t)$ just seen for the special examples of the Yule and heavy-tailed processes is easily seen to be true more generally for supercritical processes. To derive this result, begin with (\ref{eq:dde})
and observe that since $u'(s)$ is positive in the supercritical setting,  $F_t'(s)$ is an increasing function of $t$. Noting also that $F_{T-t}(0)$ is non-increasing in $t$, the integral (\ref{eq:rateeq}) must be an increasing function of $t$. That is to say, for supercritical processes,  the  reproduction rate $R_T(t)$ increases forwards in time along the ancestral lineage.
\subsection{As a function of the population size}
That the reproduction rate varies along the uniform ancestral lineage requires an explanation. Now we expand upon (\ref{heuristic}) to clarify the role of the historical population size. In the following for simplicity,  we remain with the setting $p_0=0$ so that population extinction is impossible. 

First, condition that the population size is $n$ just before time $t$. Then the total rate at which size-$\ell$ reproduction events occur in the population at time $t$ is 
      \begin{align}\label{r1}nrp_\ell.\end{align}
Next, condition on the event that at time $t$ there is a size-$\ell$ reproduction event which takes the population size from $n$ to $n+\ell-1$. Due to exchangeability, each of these $n+\ell-1$ cells are equally likely to land on the ancestral lineage. Thus one of the $\ell$ offspring of the reproduction event lands on the ancestral lineage with probability
      \begin{align}\label{p1}
\frac{\ell}{n+\ell-1}.\end{align}
Multiplying (\ref{r1}) and (\ref{p1}) gives that the rate of size-$\ell$ reproduction on the ancestral lineage at time $t$, conditioned that the population size is $n$ just before time $t$, is
     \begin{align}\label{r2}\frac{nr\ell p_\ell}{n+\ell-1}\end{align}
Observe that neither the time of sampling nor the time on the ancestral lineage play a role in (\ref{r2}) but the population size is key. Taking $n=1$ gives that the ancestral rate of size-$\ell$ reproduction is $rp_\ell$, recovering the original law in the population. On the other hand taking $n\rightarrow\infty$ the rate is $r\ell p_\ell$, which is the size-biased offspring law. Finally we note that to acknowledge the randomness of the population size one can multiple (\ref{r2}) by $\mathbb{P}[N_t=n]$ and sum over $n$ to obtain the expected rate of ancestral reproduction at time $t$ --- this calculation offers an alternative route to derive Proposition \ref{localrateprop} in the case $p_0=0$, but more importantly it points out that the time-varying law $\mathbb{P}[N_t=n]$ of the population size is the origin of variation in the ancestral reproduction rate.

In the concluding section of the paper, for a biological example of variation along sampled ancestral lineages, we are inspired by recent studies to consider the growing population of cells that is the human embryo. 

\section{Mutations in embryogenesis}

Beginning with the zygote through to the end of a human's life, cells accumulate mutations.  Most of these mutations are inconsequential to tissue function but some drive pathogenic states such as cancer. A mutation arising at an early stage in embryonic development has the potential to be especially impactful beause it may be passed on to a significant fraction of cells in fully formed tissues (see for example \cite{Wilms}), and thus a quantity relevant to human health is the mutation rate in the developing embryo.

Recent phylogenetic studies, by sequencing hundreds of genomes from multiple anatomical locations of a few adult humans, have looked along the samples' ancestral lineages backwards in time to the cell divisions which initiated these peoples' embryos \cite{Park,Coorens}.  Curiously, the studies measured that mutation rates are elevated for the earliest cell divisions.  Park and coauthors \cite{Park} estimated that the mutation rate per division was 3.8 (range 1.4 to 6.3 among five people) while the embryo size grew from one to four cells, which lowered to 1.2 (0.8 to 1.9) as the embryo grew larger.  Coorens and coauthors \cite{Coorens} estimated a mutation rate per division of 2.4 (range 1.6 to 3.2 among three people) while the embryo grew up to four cells, which subsequently dropped to 0.7 (0.5 to 1.0). Why the mutation rate should be raised at the very beginning of embryogenesis is unclear. Both studies humbly speculate that the temporary elevation could be due to initially immature DNA repair mechanisms which take time to come into action. Although their explanation seems plausible,  we suggest a more parsimonious explanation based on two facts: (1) cell division times vary (2) mutations arise not only \textit{at} but also \textit{between} cell divisions.

The point (1) on cell division rate variation is clear - cell divisions have been observed by scientists for almost two centuries, and their rates obviously vary.  We note that for mouse embryogenesis, the mean cell division time during the period of growth from 1 cell to 64 cells was calculated to be approximately 14 hours, whereas the first 2 cell divisions were estimated to be longer, each about 18 to 20 hours \cite{Mouse}.  The point (2) on mutations is only recently being understood.  It has long been believed that mutations in human tissues arise predominantly due to errors in DNA replication at cell division, but new insights from modern sequencing technologies show that many mutations do arise during the lifetime of cells independently of division, seen for example in \cite{Abascal}.  The points (1) and (2) together suggest that some cells live longer than others and so may acquire more mutations. In particular, the apparently elevated mutation rate in early embryogenesis doesn't have to be due to time-varying DNA chemistry,  instead it is plausibly just the simple consequence of relatively long cell divisions at the beginning of embryogenesis. 

\begin{figure}[h!]
\centering
\begin{tikzpicture}[xscale=1,yscale=1]
\draw[gray] (0,-3.4) -- (0,-2.4);
\draw[gray] (15,-3.4) -- (15,-2.4);

\draw[thick] [line width=1.2mm, CadetBlue] (0,-2.9) -- (15,-2.9);

\draw [very thick, black] (0.9,-2.9) node[cross] {};
\draw [very thick, black] (1.99,-2.9) node[cross] {};
\draw [very thick, black] (4.05,-2.9) node[cross] {};
\draw [very thick, black] (6.914,-2.9) node[cross] {};
\draw [very thick, black] (9.02,-2.9) node[cross] {};
\draw [very thick, black] (12.12,-2.9) node[cross] {};
\draw [very thick, black] (13.41,-2.9) node[cross] {};

\draw [fill, white] (2.8,-2.9) circle [radius=0.1];
\draw [very thick, black] (2.8,-2.9) circle [radius=0.1];

\draw [fill, white] (4.6,-2.9) circle [radius=0.1];
\draw [very thick, black] (4.6,-2.9) circle [radius=0.1];

\draw [fill, white] (7.3,-2.9) circle [radius=0.1];
\draw [very thick, black] (7.3,-2.9) circle [radius=0.1];

\draw [fill, white] (9.4,-2.9) circle [radius=0.1];
\draw [very thick, black] (9.4,-2.9) circle [radius=0.1];

\draw [fill, white] (10.5,-2.9) circle [radius=0.1];
\draw [very thick, black] (10.5,-2.9) circle [radius=0.1];

\draw [fill, white] (11.8,-2.9) circle [radius=0.1];
\draw [very thick, black] (11.8,-2.9) circle [radius=0.1];

\draw [fill, white] (12.6,-2.9) circle [radius=0.1];
\draw [very thick, black] (12.6,-2.9) circle [radius=0.1];

\draw [fill, white] (13.7,-2.9) circle [radius=0.1];
\draw [very thick, black] (13.7,-2.9) circle [radius=0.1];

\draw [fill, white] (14.3,-2.9) circle [radius=0.1];
\draw [very thick, black] (14.3,-2.9) circle [radius=0.1];

\node at (0,-3.7)   (a) {Time $0$};
\node at (15,-3.7)   (a) {Time $T$};

\end{tikzpicture}

\caption{Neutral mutations depicted with black crosses occur at a constant rate $\nu > 0$ along a sampled ancestral lineage. Cell divisions are depicted with white dots. Since the rate of divisions increases as time passes while the rate of mutations is constant, the number of mutations per division is decreasing in time.}
\label{fig:mutations}
\end{figure}
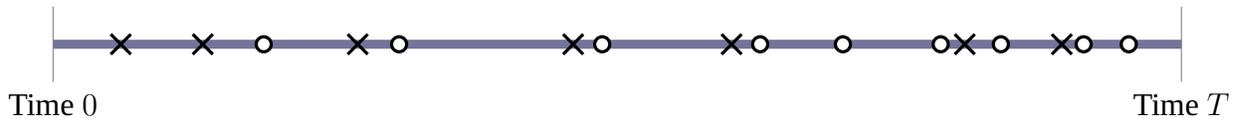

Now recalling the theme of our paper,  it is tempting to try Occam's razor once again.  The embryogenesis studies' observed that the mutation rate per division decreases along sampled ancestral lineages. Meanwhile our analysis showed that under the assumption of constant division rates in the population,  division rates increase along sampled ancestral lineages. Could a simplistic model of constant cell division and mutation rates explain the data?  Consider that cells divide at rate $r$ according to the Yule process, and that each cell acquires a Poisson$(\mu)$ number of mutations at its birth due to DNA replication errors and accumulates mutations during its lifetime as a rate-$\nu$ Poisson process. Take a uniform sample of cells from the population at a large time and look along their ancestral lineages back to the initial individual. Along any one of these lineages, the  division rate according to (\ref{Yulerate}) is $r$ at the earliest time compared to $2r$ at later times.  Due to the increasing ancestral division rate, the expected number of mutations per division decreases from $\mu+\nu\alpha^{-1}$ for the first division to $\mu+\nu \alpha^{-1}2^{-1}$ for later divisions (Figure \ref{fig:mutations}).  So this model can at most explain a two-fold decrease in the ancestral mutation rate per division, falling short of the three-fold decrease observed by \cite{Park,Coorens}.  The discrepancy may have several reasons, the most important of which we speculate to be that cell divisions do indeed speed up after the earliest divisions, accelerating not only on the ancestral lineage but on the population level too.  In any case, we find it notewothy that an incredibly simple model of constant cell division and mutation rates, with the inevitable ancestral division rate bias, can qualitatively explain \cite{Park,Coorens}'s observation.
%

Reproductive bias along sampled ancestral lineages is of course not limited to the developing embryo nor to continuous-time branching processes with homogeneous reproductive law. The concept holds for populations far more generally.  The only possible exception is a population whose genealogical structure is perfectly symmetrical, which is inconceivable in biology. 

\section*{Acknowledgements}
The second author is supported by the EPSRC funded Project EP/S036202/1 \emph{Random fragmentation-coalescence processes out of equilibrium}. The authors would like to thank Amaury Lambert for pointers on literature.

\end{document}